\documentclass{article}
\usepackage{arxiv}

\usepackage{times}

\usepackage[round]{natbib} % natbib citation style
\usepackage{xcolor} % text color
\usepackage{amsmath} % basic maths symbols 
\usepackage{amsthm} % define proof environment
\usepackage{amsfonts,amssymb} % extended maths symbols and fonts 
\usepackage{nicefrac} % compact symbols for 1/2, etc.
\usepackage{environ}
\usepackage{enumitem} % customized numbering style
\usepackage{setspace} % double space
\usepackage[ruled,linesnumbered]{algorithm2e} % algorithm package 
\usepackage{hyperref} % clickable links
\usepackage{booktabs}
\usepackage{lscape}
\usepackage{tikz} % drawing graph
\usetikzlibrary{arrows.meta} % drawing graph
\usetikzlibrary{shapes,arrows}
\usetikzlibrary{decorations.markings}
\usetikzlibrary{backgrounds}
\usetikzlibrary{positioning,chains,fit,shapes,calc}
\usetikzlibrary{graphs,quotes,trees}
\usepackage{subcaption}
\usepackage{subcaption}
\tikzset{font={\fontsize{8pt}{10}\selectfont}}

\newtheorem{lemma}{Lemma}[section]
\newtheorem{remark}{Remark}[section]

\newtheorem{definition}{Definition}[section] % definition numbers are dependent on theorem numbers
\title{Checking chordality on homomorphically encrypted graphs}

\author{
  Yang Li\\
  School of Computing\\
  The Australian National University\\
  Canberra, ACT, 2600 \\
  \texttt{kelvin.li@anu.edu.au} \\
}

\begin{document}
\maketitle

\begin{abstract}
    The breakthrough of achieving fully homomorphic encryption sparked enormous studies on where and how to apply homomorphic encryption schemes so that operations can be performed on encrypted data without the secret key while still obtaining the correct outputs. Due to the computational cost, inflated ciphertext size and limited arithmetic operations that are allowed in most encryption schemes, feasible applications of homomorphic encryption are few. While theorists are working on the mathematical and cryptographical foundations of homomorphic encryption in order to overcome the current limitations, practitioners are also re-designing queries and algorithms to adapt the functionalities of the current encryption schemes. As an initial study on working with homomorphically encrypted graphs, this paper provides an easy-to-implement interactive algorithm to check whether or not a homomorphically encrypted graph is chordal. This algorithm is simply a refactoring of a current method to run on the encrypted adjacency matrices.
\end{abstract}

\section{Introduction}

Let \textit{Alice} be the data owner of some sensitive data that is stored in a plaintext format. The principle of a cryptosystem is to securely transform Alice's data into a ciphertext so that it can be shared with others on non-secure channels and does not reveal any additional knowledge of the plaintext without access to the secret key for decryption. In addition to data sharing, if a computational agent, namely \textit{Bob}, could perform some analysis on the encrypted data without the secrete key, it would be tempting for Alice to outsource computations on the encrypted data to Bob. This ideal situation is exactly what homomorphic encryption researchers are trying to achieve. 

With homomorphic encryption (HE), basic arithmetic operations - addition, subtraction and multiplication - can be performed on encrypted data. The encrypted result can then be decrypted with the secret key to the correct result. There are, however, strict constraints on what arithmetic operations can be done and how many of them can be done before decryption fails. There have been a few studies on analysing homomorphically encrypted data to demonstrate the usefulness of HE in some real scenarios \citep{pathak2011privacy,naehrig2011can,wu2012using,graepel2012ml,bos2013improved,lauter2014private,yao2016learning}. 
That being said, all the current methods for analysing the encrypted data are simply re-engineering of the existing functions or algorithms for unencrypted data. Furthermore, multiple interactions between Alice and Bob are needed in order to remedy the limited functionality of the current HE schemes. 

As far as graphs are concerned, there have not been many studies on dealing with homomorphically encrypted graphs \citep{xie2014cryptgraph}. Obviously, graph adjacency matrices can be encrypted and simple graph statistics can be computed on the encrypted adjacency matrices, as long as the computations only involve the allowed operations, such as calculating vertex degrees.  
Based on this principle, it seems natural to separate existing graph functions or algorithms into three groups. The first group contains those that can be done naturally on encrypted graphs because they involve only HE schemes permitted operations. The second group contains those that cannot be naturally performed on HE encrypted graph data but can be re-engineered to do so. The last group contains those that cannot be done whatsoever given the limitations of the current HE development. 

As an initial investigation into this direction of research, this paper proposed a way to check the chordality of encrypted graphs.
Chordal graphs are simple but useful in many scenarios. \citet{rose1970triangulated} studied the application of chordal graphs in efficient Gaussian elimination of sparse symmetric matrices. In that, a matrix can have a corresponding chordal graph, which has a \textit{perfect elimination ordering} that corresponds to an efficient elimination scheme of the matrix. 
When doing belief propagation on graphical models using the junction tree algorithm \citep{lauritzen1988local}, the graph structures (e.g., moral graphs of Bayesian networks) need to be triangulated (i.e., made chordal) to ensure junction tress can be constructed for efficient message updating. 

Due to the limited operations that are allowed by HE schemes, the mechanism proposed here is an iterative process between Alice and Bob. The encryption step is done at Alice's side and ciphertext is shared with Bob, who then performs matrix additions and multiplications on the encrypted adjacency matrix. The output of the computation is a vector of encrypted integers, one for each graph vertex. 
Alice then decrypts, updates, encrypts and shares the vector with Bob again for the next round of computations. In the worst case, this takes at most $n$ iterations, where $n$ is the number of vertices in the graph.

\section{Chordal graph}
Throughout, the graphs considered are simple, connected and undirected. A chordal graph is also known by different names such as triangulated or recursively simplicial. Each of these names gives rise to a different but equivalent definition.

\begin{definition}
	A graph is \textbf{chordal} if every cycle of length four or more has a chord, which is an edge that is not on the cycle but connects two vertices of the cycle. 
\end{definition}

Equivalently, a graph is chordal if every induced cycle is a triangle. Hence, the name triangulated. 
Figure \ref{fig:chordal and non-chordal example} gives an example of a chordal graph and a non-chordal graph. 

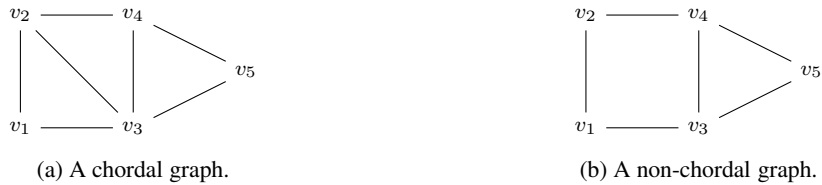
\begin{figure}[h]
	\centering
	\begin{subfigure}[b]{0.45\textwidth}
		\centering
		\begin{tikzpicture}[scale=1]
		\begin{scope}[>={Stealth[black]},every edge/.style={draw=black}]
		\node (F) at (3,0.5) {$v_2$};
		\node (G) at (4.5,0.5) {$v_4$};
		\node (H) at (3,-1) {$v_1$};
		\node (I) at (4.5,-1) {$v_3$};
		\node (J) at (6,-0.25) {$v_5$};
		\path [-] (F) edge (G);
		\path [-] (F) edge (H);
		\path [-] (G) edge (I);
		\path [-] (H) edge (I);
		\path [-] (G) edge (J);
		\path [-] (I) edge (J);
		\path [-] (I) edge (F);
		\end{scope}
		\end{tikzpicture}
		\caption{A chordal graph.}
		\label{subfig:chordal graph}
	\end{subfigure}
	\begin{subfigure}[b]{0.45\textwidth}
		\centering
		\begin{tikzpicture}[scale=1]
		\begin{scope}[>={Stealth[black]},every edge/.style={draw=black}]
		\node (F) at (3,0.5) {$v_2$};
		\node (G) at (4.5,0.5) {$v_4$};
		\node (H) at (3,-1) {$v_1$};
		\node (I) at (4.5,-1) {$v_3$};
		\node (J) at (6,-0.25) {$v_5$};
		\path [-] (F) edge (G);
		\path [-] (F) edge (H);
		\path [-] (G) edge (I);
		\path [-] (H) edge (I);
		\path [-] (G) edge (J);
		\path [-] (I) edge (J);
		\end{scope}
		\end{tikzpicture}
		\caption{A non-chordal graph.}
		\label{subfig:non chordal graph}
	\end{subfigure}
	\caption{An example of a chordal graph and a non-chordal graph.}
	\label{fig:chordal and non-chordal example}
\end{figure}

The most relevant definition to this work is recursively simplicial. 

\begin{definition}
	A vertex in a graph is called \textbf{simplicial} if its neighbours form a clique, i.e., a complete subgraph. 
\end{definition}

For example, the vertex $v_1$ is simplicial in Figure \ref{subfig:chordal graph}. The vertex $v_5$ is simplicial is both Figure \ref{subfig:chordal graph} and \ref{subfig:non chordal graph}.

Let $G=(V,E)$ be a graph with the vertex set $V$ and the edge set $E$. If a vertex is removed from a graph, the edges that are incident to the vertex are too removed. For simplicity, this is not always stated explicitly. Denote by $G\setminus \{v\}$ the subgraph of $G$ after removing the vertex $v$.  

\begin{definition}
	\label{def:rs}
	A graph $G$ is \textbf{recursively simplicial} if it has a simplicial vertex $v$ and the subgraph $G\setminus \{v\}$ is recursively simplicial, until the graph is empty.
\end{definition}

The graph in Figure \ref{subfig:chordal graph} is recursively simplicial. It has a simplicial vertex $v_5$ and the subgraph after removing $v_5$ (and the edges $\{v_3 v_5, v_4 v_5\}$) is again recursively simplicial. Repeatedly, the original graph is eliminated to the empty graph.  
The graph in Figure \ref{subfig:non chordal graph}, however, is not recursively simplicial. After removing the only simplicial vertex $v_5$ from it, the remaining graph is a square that has no simplicial vertex, hence not recursively simplicial. 

Definition \ref{def:rs} is a constructive definition that leads to a vertex ordering, following which a chordal graph can be eliminated completely. This vertex ordering is called a \textit{perfect elimination ordering} of the graph. There may be more than one perfect elimination ordering in a graph (as in the case of Figure \ref{subfig:chordal graph}). A more in-depth review of some studies of chordal graphs and related characteristics is in \citep{heggernes2006minimal}. 

From more than three decades ago, it was known that chordal graphs can be recognized in linear time \citep{rose1976algorithmic,tarjan1984simple}. This paper, however, does not refactor the algorithm proposed by \cite{rose1976algorithmic}. Their algorithm is based on the lexicographic breath-first search and produces a lexicographic ordering of the vertices, the reverse of which is a perfect elimination ordering if and only if the graph is chordal. The algorithm was designed so that it is easy to be implemented on computers. But it involves comparisons of indices and sets that are not supported by most (if not all) HE schemes. 

Instead, the strategy adopted here (Algorithm \ref{alg:check chordal}) is based on the constructive definition of recursively simplicial. It runs in polynomial time and can be easily performed with only additions and multiplications of graph adjacency matrices and some help from Alice as shown later in Section \ref{sec:he check chordal}. 

\begin{figure}[ht]
    \centering
    \begin{minipage}{0.5\textwidth}
        \begin{algorithm}[H]
             \SetKwInOut{Input}{Input}
             \SetKwInOut{Output}{Output}
             \Input{a graph $G$}
             \Output{TRUE or FALSE}
             \While{exists a simplicial vertex $v$}{
                $G=G\setminus \{v\}$
             }
             
             \eIf{$G=\emptyset$}{
               return TRUE\;
               }{
               return FALSE\;
              }
         \caption{Checking chordality}
         \label{alg:check chordal}
        \end{algorithm}
    \end{minipage}
\end{figure}

\begin{remark}
\label{remark:sim nodes order}
Since there may be more than one simplicial vertex in a graph, an important remark of Algorithm \ref{alg:check chordal} is that the algorithm is correct regardless of the order of the simplicial vertices been removed and the number of simplicial vertices been removed at once. For the rigorous proof of this, the readers are suggested to read the survey \citep{heggernes2006minimal}.
\end{remark}

\section{Homomorphic encryption}

Shortly after the well known RSA encryption scheme \citep{rivest1978method} was released, \cite{rivest1978data} raised the question of whether it is possible to perform arithmetic operations (e.g., addition and multiplication) on encrypted data without the secret key, so that the results can be decrypted to the correct results if the same operations were performed on the plaintext. An encryption scheme possessing such a property is called a homomorphic encryption scheme. It can be formally stated as follows (under the public key schemes). 

\begin{definition}
Let $m_1$ and $m_2$ be two plaintexts, $pk$ and $sk$ be the public key and secret key for encryption and decryption, respectively. If an encryption scheme satisfies that for an operation $\circ_m$ in the plaintext space, there is a corresponding operation $\circ_c$ in the ciphertext space such that  
    \begin{equation}
    \label{eq:he}
        Dec(sk,Enc(pk,m_1) \circ_c Enc(pk,m_2)) = m_1 \circ_m m_2,
    \end{equation}
    then the encryption scheme is said to be \textbf{homomorphic}. 
\end{definition}

Most of the HE schemes have the same operations in both plaintext and ciphertext spaces. That is, additions of ciphertexts can be decrypted to the addition of plaintexts. So is true for multiplications. 
The name of ``homomorphic'' may have been taken from \textit{homomorphism} in mathematics. The analogy is that decryption is a homomorphism from the ciphertext space to the plaintext space that preserves the operations of the spaces as stated in Equation \ref{eq:he}. The process is also shown in a diagram in Figure \ref{fig:he}. 

% We need layers to draw the block diagram
\pgfdeclarelayer{background}
\pgfdeclarelayer{foreground}
\pgfsetlayers{background,main,foreground}

% Define a few styles and constants
%\tikzstyle{sensor}=[draw, fill=blue!20, text width=5em, text centered, minimum height=2.5em]
%\tikzstyle{ann} = [above, text width=5em]
%\tikzstyle{naveqs} = [sensor, text width=6em, fill=blue!5, minimum height=5em, rounded corners]
%\def\blockdist{2.3}
%\def\edgedist{2.5}

\begin{figure}[h]
    \centering
    \begin{tikzpicture}[scale=0.9]
	\node (raw) [draw, text width=2em, fill=blue!20, minimum height=2.5em, text centered, rounded corners] {$A$};
    % Note the use of \path instead of \node at ... below. 
    % alice
    \path (raw)+(3.5,0) node (cipher1) [draw, fill=blue!5, text width=4em, 
    text centered, minimum height=2.5em, rounded corners] {$\tilde{A}$};
    % bob
    \path (raw)+(9,0) node (cipher2) [draw, fill=blue!7, text width=4em, 
    text centered, minimum height=2.5em, rounded corners] {$\tilde{A}$};
    \path (cipher2)+(0,-3) node (cipher3) [draw, fill=yellow!15, text width=8em, 
    text centered, minimum height=2.5em, rounded corners] {$\tilde{M}=(\tilde{A}*\tilde{A})\cdot \tilde{A}$};
    % alice
    \path (raw)+(0,-3) node (raw2) [draw, fill=blue!20, text width=2em, 
    text centered, minimum height=2.5em, rounded corners] {$M$};
    \path (raw2)+(3.5,0) node (cipher4) [draw, fill=blue!5, text width=4em, 
    text centered, minimum height=2.5em, rounded corners] {$\tilde{M}$};
    
    % Unfortunately we cant use the convenient \path (fromnode) -- (tonode) 
    % syntax here. This is because TikZ draws the path from the node centers
    % and clip the path at the node boundaries. We want horizontal lines, but
    % the sensor and naveq blocks aren't aligned horizontally. Instead we use
    % the line intersection syntax |- to calculate the correct coordinate
    \path [draw, ->] (raw) -- node [above] {Enc(pk)} (cipher1.west |- raw);
    \path [draw, ->] (cipher1) -- node [above,pos=0.41] {Send} (cipher2.west |- cipher1);
    \path [draw, ->] (cipher2) -- node [left] {Evaluate} (cipher2.south |- cipher3.north);    
    \path [draw, ->] (cipher3) -- node [above] {Send} (cipher4.east |- cipher3);
    \path [draw, <-] (raw2) -- node [above] {Dec(sk)} (cipher4.west |- raw2);
    \path [draw, ->, dashed] (raw) -- node [left,text width=7em] {If Alice had the computing power} (raw.south |- raw2.north);    
 
    \path (raw.north)+(-1.1,1.1) node (alice) {\large Alice - data owner};
    \path (cipher2.north)+(0.5,1.1) node (bob) {\large Bob - computational agent}; 
    
    % Now it's time to draw the colored IMU and INS rectangles.
    % To draw them behind the blocks we use pgf layers. This way we  
    % can use the above block coordinates to place the backgrounds   
    \begin{pgfonlayer}{background}
        % Compute a few helper coordinates
        % big yellow background
        \path (raw.west |- raw.north)+(-2.5,1.5) node (a) {};
        \path (raw2.south -| cipher1.east)+(+0.5,-0.5) node (b) {};
        \path[fill=green!20,rounded corners, draw=black!50, dashed]
            (a) rectangle (b);
                  
        \path (cipher3.west |- raw.north)+(-0.5,1.5) node (a) {};
		\path (raw2.south -| cipher3.east)+(+2,-0.5) node (b) {};
		\path[fill=red!20,rounded corners, draw=black!50, dashed]
			(a) rectangle (b);      
        
    \end{pgfonlayer}
    
\end{tikzpicture}
    \caption{Diagram of a homomorphic encryption scheme. The data owner Alice encrypted the data $A$ using the public key and sent the ciphertext $\tilde{A}$ to the computational agent Bob, who then performed HE permitted operations on the ciphertext. We denote by $*$ and $\cdot$ the matrix multiplication and matrix element-wise multiplication, respectively. The resulting ciphertext $\tilde{M}$ was sent back to Alice for decryption using the secret key. Alice could obtain the same result $M$ from $A$ if she had the computing power to perform the same operations that Bob did on the ciphertext.}
    \label{fig:he}
\end{figure}
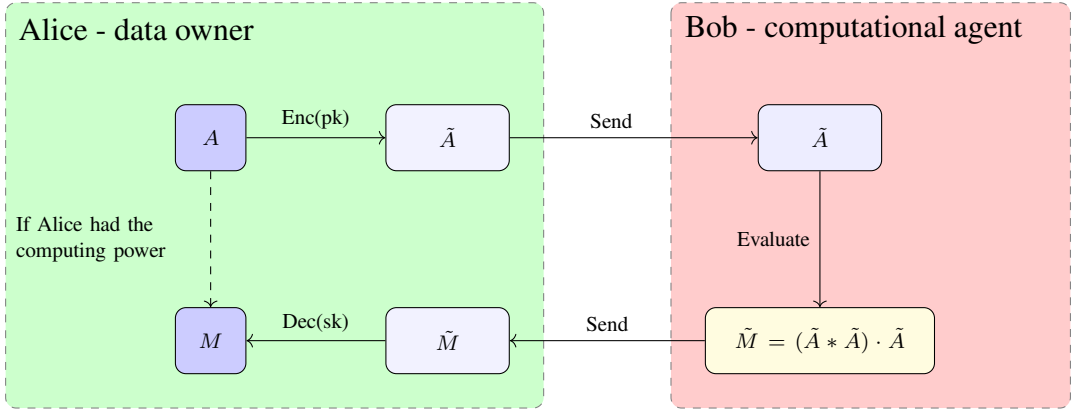

It is important to note, however, that due to \textit{semantic security} that is required by HE, encryption is not homomorphic, i.e.,  
\begin{equation*}
    Enc(pk,m_1) \circ_c Enc(pk,m_2)) \neq Enc(pk,m_1 \circ_m m_2),
\end{equation*}
because encryptions in HE are probabilistic instead of deterministic. Simply speaking, semantic security states that given two messages $m_1$ and $m_2$ and a ciphertext $c$ that encrypts either one of them, it is ``hard'' to decide which one $c$ encrypts. To achieve semantic security, an HE scheme often injects a small amount of random noise into the ciphertext so that there are many choices of ciphertext for encrypting a plaintext. 

Many encryption schemes have been shown to be able to handle either addition or multiplication, but not both. And only a bounded number of these operations can be applied before decryption fails. The breakthrough of achieving an unlimited number of operations of both additions and operations, known as fully homomorphic encryption (FHE), was made by \cite{gentry2009fully}. Since then, there have been some great works that advanced the field by proving new security guarantees or efficient schemes for FHE.   
%This paper uses the BFV scheme \citep{fan2012somewhat} as the underlying encryption scheme for the analysis. Other schemes can be used too. 

\section{Check chordality on encrypted graphs} 
\label{sec:he check chordal}
This section presented the details of how to re-engineer Algorithm \ref{alg:check chordal} so that it uses only HE allowed operations on encrypted adjacency matrices.

\begin{figure}[ht]
    \centering
    \begin{tikzpicture}[scale=1]
		\begin{scope}[>={Stealth[black]},every edge/.style={draw=black}]
		\node (F) at (3,0.5) {$v_2$};
		\node (G) at (4.5,0.5) {$v_4$};
		\node (H) at (3,-1) {$v_1$};
		\node (I) at (4.5,-1) {$v_3$};
		\node (J) at (6,-0.25) {$v_5$};
		\path [-] (F) edge (G);
		\path [-] (F) edge (H);
		\path [-] (G) edge (I);
		\path [-] (H) edge (I);
		\path [-] (G) edge (J);
		\path [-] (I) edge (J);
		\path [-] (I) edge (F);
		\end{scope}
		\end{tikzpicture}
    \caption{A chordal graph.}
    \label{fig:chordal graph 2}
\end{figure}
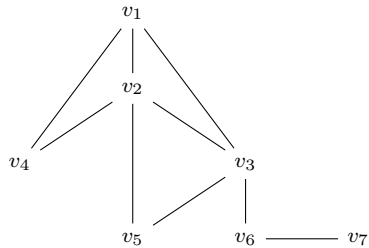

Given the graph in Figure \ref{fig:chordal graph 2}, its adjacency matrix is
\begin{equation*}
    A=
    \begin{pmatrix}
    0 & 1 & 1 & 0 & 0\\
    1 & 0 & 1 & 1 & 0\\
    1 & 1 & 0 & 1 & 1\\
    0 & 1 & 1 & 0 & 1\\
    0 & 0 & 1 & 1 & 0
    \end{pmatrix},
\end{equation*}
where its rows and columns have labels $\{v_1, v_2, v_3, v_4, v_5\}$.
Multiplying the adjacency matrix by itself gives another matrix
\begin{equation*}
    A_2=A*A=
    \begin{pmatrix}
    2 & 1 & 1 & 2 & 1\\
    1 & 3 & 2 & 1 & 2\\
    1 & 2 & 4 & 2 & 1\\
    2 & 1 & 2 & 3 & 1\\
    1 & 2 & 1 & 1 & 2
    \end{pmatrix}.
\end{equation*}
Each entry $A_2[i,j]$ corresponds to the number of 2-paths from vertex $i$ to $j$ which may or may not be 0. Multiplying the matrix $A$ and $A_2$ element-wise gives another matrix 
\begin{equation*}
    M=A\cdot A_2=
    \begin{pmatrix}
    0 & 1 & 1 & 0 & 0\\
    1 & 0 & 2 & 1 & 0\\
    1 & 2 & 0 & 2 & 1\\
    0 & 1 & 2 & 0 & 1\\
    0 & 0 & 1 & 1 & 0
    \end{pmatrix}.
\end{equation*}
In each entry $M[i,j]$, if $j$ is a neighbour of the vertex $i$, then it corresponds to the number of 2-paths from vertex $i$ to $j$ which again may or may not be 0. If $j$ is not a neighbour of $i$, the entry is always 0. 

The intuition behind these multiplications is that if two vertices $u$ and $v$ in a graph $G$ is both connected by a 1-path and a 2-path, then the vertices $u,v$ and a third vertex form a triangle. Therefore, identifying a simplicial vertex from a graph becomes identifying a vertex, which forms a triangle with each pair of its neighbours. In other words, together they form a clique. 

The graph $G$ in Figure \ref{fig:chordal graph 2} has two simplicial vertices, $v_1$ and $v_5$. It is noticeable that in the matrix $M$, the rows and columns correspond to these two simplicial vertices are quite unique, comparing with other rows and columns. This is stated in the following lemma. 

\begin{lemma}
\label{lemma:suf and nec condi for sim nodes}
Given a graph $G = (V,E)$, a vertex $v$ with degree $d(v)$ in $G$ is simplicial if and only if its corresponding row $M[v,]$ and column $M[,v]$ in the matrix $M$ satisfy the following: 
\begin{enumerate}
    \item there are $d(v)$ non-zero entries in $M[v,]$ and $M[,v]$ and 
    \item all these non-zero entries have the same value $d(v)-1$.
\end{enumerate}
\end{lemma}

\begin{proof}
The lemma is proved for the row vector, but the same argument applies to the column vector too.
Suppose the vertex $v$ is simplicial in $G$.
For each vertex $u \in N(v)$, there are $d(v)-1$ 2-paths from $v$ to $u$, because a 2-path from $v$ to $u$ must go through a neighbour of $v$ and vertices in $N(v)$ are pairwise adjacent due to $v$'s simpliciality. 
Therefore, in each entry $M[v,j]$, if $j$ is a neighbour of the vertex $v$, then the entry has value $d(v)-1$. 

Suppose both conditions are satisfied. 
To prove by contradiction, assume $v$ is not a simplicial vertex in $G$, i.e., there exists vertices $j,k \in N(v)$ such that $jk \notin E$ is not an edge in $G$. 
But both $j$ and $k$ are neighbours of $v$ so their entries $A[v,j]$ and $A[v,k]$ are non-zero. If $v$ has degree $d(v)=2$, then $A_2[v,j]=A_2[v,k]=0$ because there is no 2-path from $v$ to any of them. Hence, after multiplying the two matrices element-wise, it gives $M[v,j]=M[v,k]=0$ which contradicts with condition 1. 
If the degree $d(v) \ge 3$, then the entries $A_2[v,j]$ and $A_2[v,k]$ are non-zero. 
However, since the edge $jk$ is missing from the graph $G$, there are at most $d(v)-2$ 2-paths from the vertex $v$ to $j$ or $k$. So the entries $A_2[v,j]$ and $A_2[v,k]$ are at most $d(v)-2$, which contradicts with condition 2.
Therefore, the vertex $v$ must be simplicial in $G$.
\end{proof}

By using this lemma, one is able to run on a homomorphically encrypted adjacency matrix to look for simplicial vertices with the following operation. For each vertex $i$, calculate 
\begin{equation}
\label{eq:he find sim}
    c=f(i)=\sum_{j=1}^n M[i,j] - \sum_{j=1}^n A[i,j] \cdot \left(\sum_{j=1}^n A[i,j]-1\right).
\end{equation}
Once the calculation is done, the user will receive the encrypted result $c$. Decrypting $c$ with the secret key will give an answer of either 0 or a non-zero integer, indicating the vertex $i$ is simplicial or not, respectively. As there may be more than one simplicial vertex in each graph, Equation \ref{eq:he find sim} can be applied to all vertices in the graph simultaneously. The resulting ciphertext will be a vector, decrypting which will give a vector of integers, where 0 indicates the corresponding vertex is simplicial.   

\begin{figure}[ht]
    \centering
    \begin{tikzpicture}[scale=1]
		\begin{scope}[>={Stealth[black]},every edge/.style={draw=black}]
		
		\node (4) at (0,0) {$v_4$};
		\node (2) at (1.5,1) {$v_2$};
		\node (1) at (1.5,2) {$v_1$};
		\node (5) at (1.5,-1) {$v_5$};
		\node (3) at (3,0) {$v_3$};
		\node (7) at (4.5,-1) {$v_7$};
		\node (6) at (3,-1) {$v_6$};
	
		\path [-] (4) edge (1);
		\path [-] (4) edge (2);
		\path [-] (2) edge (1);
		\path [-] (2) edge (5);
		\path [-] (1) edge (3);
		\path [-] (2) edge (3);
		\path [-] (5) edge (3);
		\path [-] (3) edge (6);
		\path [-] (6) edge (7);
		\end{scope}
		\end{tikzpicture}
    \caption{A chordal graph.}
    \label{fig:check sim nodes}
\end{figure}

For example, given the graph $G$ in Figure \ref{fig:check sim nodes}, its corresponding matrices are shown below with the row and column labels $(v_1,v_2,v_4,v_3,v_5,v_6,v_7)$.
\begin{align*}
    A=
    \begin{pmatrix}
    0 & 1 & 1 & 1 & 0 & 0 & 0\\
    1 & 0 & 1 & 1 & 1 & 0 & 0\\
    1 & 1 & 0 & 0 & 0 & 0 & 0\\
    1 & 1 & 0 & 0 & 1 & 1 & 0\\
    0 & 1 & 0 & 1 & 0 & 0 & 0\\
    0 & 0 & 0 & 1 & 0 & 0 & 1\\
    0 & 0 & 0 & 0 & 0 & 1 & 0
    \end{pmatrix},
    A_2=
    \begin{pmatrix}
    3 & 2 & 1 & 1 & 2 & 1 & 0\\
    2 & 4 & 1 & 2 & 1 & 1 & 0\\
    1 & 1 & 2 & 2 & 1 & 0 & 0\\
    1 & 2 & 2 & 4 & 1 & 0 & 1\\
    2 & 1 & 1 & 1 & 2 & 1 & 0\\
    1 & 1 & 0 & 0 & 1 & 2 & 0\\
    0 & 0 & 0 & 1 & 0 & 0 & 1
    \end{pmatrix},
    M=
    \begin{pmatrix}
    0 & 2 & 1 & 1 & 0 & 0 & 0\\
    2 & 0 & 1 & 2 & 1 & 0 & 0\\
    1 & 1 & 0 & 0 & 0 & 0 & 0\\
    1 & 2 & 0 & 0 & 1 & 0 & 0\\
    0 & 1 & 0 & 1 & 0 & 0 & 0\\
    0 & 0 & 0 & 0 & 0 & 0 & 0\\
    0 & 0 & 0 & 0 & 0 & 0 & 0
    \end{pmatrix}.
\end{align*}
By applying Equation \ref{eq:he find sim} for each vertex $i$ repeatedly, the output is a vector $\Vec{s}=(-2,-6,0,-8,0,-2,0)$ which indicates the vertices $v_4,v_5,v_7$ are simplicial in the current graph. This example demonstrates how to use Equation \ref{eq:he find sim} to find all simplicial vertices at once. In reality, the adjacency matrix $A$ is encrypted homomorphically, so are $A_2$, $M$ and $\Vec{s}$. Hence, the outputs of all operations are ciphertext and make no sense to Bob. 

The time complexity of the matrix multiplication $A*A$ is $O(n^3)$ if directly applying the mathematical definition of matrix multiplication. The element-wise multiplication takes $O(n^2)$ time. So it takes $O(n^3)$ to get the matrix $M$. To test if a vertex is simplicial using Equation \ref{eq:he find sim}, each summation takes time $O(n)$. If $\sum_{j=1}^n A[i,j]$ can be stored in a variable, then it takes $O(n)$ time to tell if a vertex is simplicial or not and $O(n^2)$ time to find all simplicial vertices at once.  
Therefore, the entire process of finding simplicial vertices takes $O(n^3)$ time. The time complexity can be reduced by using more efficient matrix computations, but this is not in the scope of this paper. 

Unfortunately, due to the constraint of not allowing division and comparison operations, there is no way for Bob to recursively delete and check simplicial vertices based on the encrypted adjacency matrix. To be more precise, deleting simplicial vertices can be done by multiplying each row and column of the adjacency matrix by the corresponding value in the output vector, i.e., 
\begin{align}
\label{equ:delete sim nodes}
    A[i,] &= A[i,]\cdot \Vec{s}[i], \nonumber \\
    A[,i] &= A[,i]\cdot \Vec{s}[i],
\end{align}
This way the corresponding rows and columns of simplicial vertices become 0, which is equivalent to remove these vertices from the current graph. 

The difficulty going in this direction is that the corresponding rows and columns of the non-simplicial vertices will have integer values no less than 1. In other words, the resulting adjacency matrix represents a weighted graph with integer weights taken from $[1,\infty)$, so Lemma \ref{lemma:suf and nec condi for sim nodes} is not applicable in this case. There are two possibilities (perhaps more). First, generalize Lemma \ref{lemma:suf and nec condi for sim nodes} to weighted graphs. The difficulty is in condition 2, assuming a vertex is simplicial, to calculate the exact value of each entry in the matrix $M$ requires the knowledge of which edge is missing from the current graph because edge weights are different. But this is unknown to Bob. Second, reduce all non-zero values in the output vector $\Vec{s}$ to 1, so that when applying Equation \ref{equ:delete sim nodes} the resulting adjacency matrix consists of only 0s and 1s. This, however, seems to be a mission impossible. Because even for a length 2 vector $\Vec{s}=(s_1,s_2)$ where exactly one of the element is 0, it seems impossible without knowing which element is 0 to transfer $\Vec{s}$ to a unit integer vector $(1,0)$ or $(0,1)$ with only additions, subtractions and multiplications.

Because of the limited operations allowed in the current HE schemes, once the vector $\Vec{s}$ is updated by Bob, he needs to send $\Vec{s}$ to Alice for decryption, updating and encryption, so that Bob can carry out the next round of computations with the updated $\Vec{s}$ using Equation \ref{eq:he find sim} and Equation \ref{equ:delete sim nodes}. In this process, all Alice needs to do is to decrypt $\Vec{s}$ with the secret key and replace the non-zero elements in $\Vec{s}$ with 1s. Once done, all elements in the updated $\Vec{s}$ will be re-encrypted and shared with Bob again without revealing any additional information about the graph because of HE's semantic security.\footnote{If Alice only re-encrypts the non-zero elements and uses the same encryptions for the zero elements, Bob will be able to tell which vertices were re-encrypted and which are not.} With the new $\Vec{s}$, Bob will then apply Equation \ref{equ:delete sim nodes} to the current adjacency matrix $A$ to ``remove'' simplicial vertices from the current graph. What is certain is that Bob has no knowledge of which vertices are removed from the current graph due to their simpliciality.

The interactive process between Alice and Bob is repeated several times until one of the following stopping conditions is satisfied.
\begin{enumerate}
    \item The vector $\Vec{s}=0$ is a zero vector. It indicates all vertices have been deleted from the original graph, so the original graph is chordal. 
    \item The non-zero vector $\Vec{s}$ does not change in two consecutive rounds. It means no vertices can be further deleted from the current graph, so the original graph is not chordal. 
\end{enumerate}
In either case, Alice can terminate the process and obtain the final decision on the given graph. The maximum number of interactions between Alice and Bob is $n$, where $n$ is the number of vertices in the graph. This corresponds to the case where the graph is chordal and only one vertex can be removed in each iteration. The pseudocode of this process is presented in Algorithm \ref{alg:check chordal on enc graphs}.

\begin{figure}[ht]
    \centering
    \begin{minipage}{0.7\textwidth}
    
        \begin{algorithm}[H]
             \SetKwInOut{Input}{Input}
             \SetKwInOut{Output}{Output}
             \Input{an $n$ by $n$ graph adjacency matrix $A$, public key $pk$ and secret key $sk$}
             \Output{TRUE or FALSE}

             % Alice encrypts A and s
             $A=Enc(pk,A)$ \tcp*{Encrypt $A$}
             
             $\Vec{s} = (1, \dots, 1)$ \\
             $k'=n+1$\\
             
             \While{$\Vec{s}\neq 0$ and $\sum(\Vec{s})<k'$}{

             $k'=\sum \Vec{s}$\\
             $\Vec{s}=Enc(pk,\Vec{s})$ \tcp*{Encrypt $\Vec{s}$}

                % update A
                \For{$i\gets 1$ \KwTo $n$}{
                    $A[i,]=A[i,] \cdot \Vec{s}[i]$ \tcp*{Update $A$}
                    $A[,i]=A[,i] \cdot \Vec{s}[i]$
                }
                
                $A_2=A*A$ \\
                $M=A_2 \cdot A$\\
                
                \For{$i\gets 1$ \KwTo $n$}{
                    $\Vec{s}[i]=f(i)$ \tcp*{Equation \ref{eq:he find sim}}
                }
                
                % decrypt s
                $\Vec{s}=Dec(sk,\Vec{s})$ \tcp*{Decrypt $\Vec{s}$}
                
                \For{$i\gets 1$ \KwTo $n$}{
                    \uIf{$\Vec{s}[i] \neq 0$}{
                        $\Vec{s}[i]=1$ \tcp*{Update $\Vec{s}$}
                    }
                }
                
                %$k'=\sum \Vec{s}$\\
                %$\Vec{s}=Enc(pk,\Vec{s})$ \tcp*{Encrypt $\Vec{s}$}
                
             }
             
             \eIf{$\Vec{s}=0$}{
               return TRUE\;
               }{
               return FALSE\;
              }
         \caption{Checking chordality on homomorphically encrypted graphs}
         \label{alg:check chordal on enc graphs}
        \end{algorithm}
    \end{minipage}
\end{figure}

The algorithm starts with Alice encrypting the adjacency matrix $A$ and initializing the vector $\Vec{s}$ for Bob to mark which vertices are simplicial. The \textit{While} loop is controlled by Alice as explained above in the two stopping conditions. Inside the \textit{While} loop, Bob performs the arithmetic operations on the ciphertext $A$ to update the vector $\Vec{s}$ that is also a ciphertext. Once Bob completes a round of computations, Alice decrypts $\Vec{s}$ and updates it in order to decide whether or not the algorithm should proceed.

\section{Conclusion}
This paper re-factored an algorithm for checking graph chordality on homomorphically encrypted graphs. The re-factored algorithm runs on encrypted graph adjacency matrices with only addition, subtraction and multiplication operations, but it requires multiple rounds of communications between the data owner Alice and the computational agent Bob to iteratively remove vertices from the current graph in an encrypted manner. In addition, because the current HE schemes only allow certain operations to be performed on encrypted data, the running time of the algorithm in each interaction is $O(n^3)$ which is slower than a linear time algorithm on the unencrypted graphs.

\section*{Acknowledgement}
I thank Dr Kee Siong Ng (The Australian National University) and Dr Lloyd Allison (Monash University) for useful discussions.  

\newpage 
\bibliography{references}
\bibliographystyle{abbrvnat}
\end{document}